\documentclass[journal]{IEEEtran}
\usepackage[cmex10]{amsmath}
\usepackage{amssymb}
\usepackage{amsthm}
\usepackage{epsfig}
\usepackage{color}
\usepackage{algorithm}
\usepackage{algorithmic}
\usepackage{multirow}
\usepackage{array}
\newtheorem{theorem}{Theorem}

\usepackage{epsfig}
\usepackage{color}

\newcommand{\squeezeup}{\vspace{-2.5mm}}
\begin{document}
\title{One Breaker is Enough: \\Hidden Topology Attacks on Power Grids}
\author{\authorblockN{Deepjyoti Deka,~ Ross Baldick ~and~ Sriram Vishwanath\\}
\authorblockA{Department of Electrical \& Computer Engineering\\
University of Texas at Austin\\
Email: deepjyotideka@utexas.edu, baldick@ece.utexas.edu, sriram@ece.utexas.edu }}
\maketitle
\begin{abstract}
A coordinated cyber-attack on grid meter readings and breaker statuses can lead to incorrect state estimation that can subsequently destabilize the grid. This paper studies cyber-attacks by an adversary that changes breaker statuses on transmission lines to affect the estimation of the grid topology. The adversary, however, is incapable of changing the value of any meter data and can only block recorded measurements on certain lines from being transmitted to the control center. The proposed framework, with limited resource requirements as compared to standard data attacks, thus extends the scope of cyber-attacks to grids secure from meter corruption. We discuss necessary and sufficient conditions for feasible attacks using a novel graph-coloring based analysis and show that an optimal attack requires breaker status change at only ONE transmission line. The potency of our attack regime is demonstrated through simulations on IEEE test cases.
\end{abstract}

\section{Introduction}
Real time operation of the power grid and computation of electricity prices \cite{LMP} require accurate estimation of its structure and critical state variables. Remote Terminal Units (RTUs) transmit measurements collected from different grid components to the central control center for state estimation and subsequent use in analyzing grid stability. The collected measurements can be broadly classified into two kinds: meter readings and breaker statuses. The breaker statuses on transmission lines help create the current operational topology of the grid. The meter readings, comprising of line flow and bus power injection measurements, are then used to estimate the state variables over the estimated topology. In a practical setting, the collected measurements suffer from noise, that get added at source or during communication to the control center. The affect of such noise is minimized through placement of redundant/additional meters and use of suitable bad-data detection and correction techniques at the estimator \cite{estimation}.
%Flow measurements on lines and power injections on buses constitute the available meter readings while transmission line breaker statuses provide topological measurements. It is worth noting that while meter readings are real-valued, breaker statuses are binary in nature with $0$ denoting an open line and $1$ denoting a closed/operation line.

Cyber-attacks on the power grid refer to corruption of measurements (meter readings and breaker statuses) by an adversary, aimed at changing the state estimation output, without getting detected by the estimator's checks. The viability of such attacks has in fact been demonstrated through controlled experiments like the Aurora attack in Department of Energy's Idaho Laboratory and GPS spoofing attack on phasor measurement units (PMUs) \cite{gps}. Past literature on cyber-attacks have generally looked at adversaries that change meter data (and not breaker statuses) to affect state estimation. Such data attacks involving injection of malicious data into meters were first analyzed in \cite{hidden}. Using a DC power flow model for state estimation, the authors of \cite{hidden} provide an attack design using projection matrices. Following this, several approaches have been discussed to study hidden attacks under different operating conditions. These include mixed integer programming \cite{sou}, heuristic based detection \cite{thomas}, sparse recovery using $l_0-l_1$ relaxation \cite{poor}, graph-cut based construction for systems with phasor measurement units (PMUs) \cite{deka2} among others. The possible economic ill-affects of such hidden data attacks on power markets are presented in \cite{price}. In a recent paper \cite{topology1}, the authors investigates hidden attacks under the more general and potent regime of topology data (breaker statuses) and meter data corruption. All of these cited work on 'data alone' or 'topology and data' attacks, however, require changing floating point meter measurements in real time. The practicality of this is questionable as significant resources are required to synchronize the changes at multiple meters.

In this paper, we focus on hidden attacks that primarily operate through changes in breaker statuses. Here the adversary changes the statuses of a few operational breakers from $1$ (closed) to $0$ (open), as well as jams (blocks the communication) of flow measurements on a subset of transmission lines in the grid. However, the adversary does not modify any meter reading to an arbitrary value. We term these attacks as 'breaker-jammer' attacks. Note that breaker statuses, unlike meter readings,  are binary in nature and fluctuate with lower frequency. They are thus easier to change, even by adversaries with limited resources. Jamming measurements, through jammers or by destruction of communication apparatus, is technologically less intensive than corrupting meter measurements. In fact, jamming does not raise a major alarm as measurement loss due random communication drops occurs under normal circumstances. The 'breaker-jammer' attack model was introduced by the authors for grids with a specific meter configuration requiring sufficient line flow measurements in \cite{infocom}. This work generalizes the framework to any grid with line flow and injection meters and uses a novel graph-coloring analysis to determine the optimal hidden attack. Our graph coloring based analysis is in principal similar to \cite{poona} which studies standard data attacks as a graph partitioning problem. However, the similarly ends there as our attack model does not use corruption of meter readings. Instead breaker status changes and line flow jams provide a different set of necessary and sufficient conditions for feasible attacks. The surprising revelation of our analysis is that under normal operating conditions, a single breaker status change (with the necessary flow measurement jamming) is sufficient to create an undetectable attack. In fact, we show that if a hidden attack can be constructed by changing the status of a set of breakers, then a hidden attack using only one break status change exists as well. This is significant as the adversary can focus on jamming the necessary flow measurements, after selecting a breaker to attack. Further, our attack design does not depend on the current system state or transmission line parameter values, and has low information requirements.

The rest of this paper is organized as follows. We present the system model used in generalized state estimation and describe the attack model in the next section. The graph coloring approach to determine the necessary and sufficient conditions for a hidden attack and elucidating examples are discussed in Section \ref{sec:coloring}. The design of the optimal hidden attack is discussed in Section \ref{sec:design} along with simulations on IEEE test cases. Finally, concluding remarks and future directions of work are presented in Section \ref{sec:conclusion}.

\section{Generalized State Estimation in the Power Grid and Topology-based Attacks}
\label{sec:attack}
First, we provide a brief description of the notation used. We represent the current operational structure of the grid by graph $\cal G = (\cal V, \cal E)$ where $\cal V$ denotes the set of buses/nodes of size $n_{\cal V}$ and $\cal E$  denotes the set of operational edges of size $n_{\cal E}$. The set of binary breakers statuses for the edges is denoted by the diagonal matrix $D$ of size $n_{\cal E} \times n_{\cal E}$. We assume that all lines to be initially operational ($D$ is identity matrix) and ignore any non-operation line for ease of notation. The edge to node incidence matrix is denoted by $M$ of dimension $n_{\cal E} \times n_{\cal V}$. Each operational edge $(ab)$ between nodes $a$ and $b$ has a corresponding row $M_{ab}$ in $M$, where $M_{ab} = e_a'-e_b'$. $e_a$ denotes the standard basis vector in $\mathbb{R}^{n_{\cal V}}$ with one at the $a^{th}$ location. The direction of flow on edge $(ab)$ is taken to be from $a$ to $b$, without any loss of generality. We consider the DC power flow model for state estimation in this paper \cite{bookdc}. The state variables in this model are the bus phase angles, denoted by the $n_{\cal V}\times 1$ vector $x$. The set of measurements is denoted by the vector $z = \setlength{\arraycolsep}{2pt} \renewcommand{\arraystretch}{0.8}\begin{bmatrix} z_f \\z_{inj} \end{bmatrix}$. Here line flow measurements are included in $z_f$ and bus injection measurements are included in $z_{inj}$.

State estimation in the power grid relies on the breaker statuses in $D$ for topology estimation and then uses the meter measurements $z$ for estimating the state vector $x$. The relation between $x$ and $z$ in the DC model is given by $z =  Hx + e$ where $e$ is the zero mean Gaussian noise vector with covariance matrix $\Sigma_e$. $H$ is the measurement matrix and depends on the grid structure and susceptance of transmission lines. Let the $k_1^{th}$ entry in $z$ corresponds to the flow measurement on line $(ab)$. Then $H(k_1,:)$ (the $k_1^{th}$ row in $H$) is given by
\begin{align}
H(k_1,:) = [0...B_{ab}..0..-B{ab}..0] = B_{ab}M_{ab} \label{flow}
\end{align}
with the non-zero values at the $a^{th}$ and $b^{th}$ locations respectively. $B_{ab}$ is the susceptance of the line $(ab)$. On the other hand, if the $k_2^{th}$ entry corresponds to an injection measurement at node $a$, we have $z(k_2) = \sum_{r: (ar) \in \cal E}B_{ar}(x_a-x_r)$. In matrix form, ignoring measurement noise, we can write equations for received measurements as
\begin{align}
z_f = TBMx \text{~~for flow measurements} \label{flowmat}\\
z_{inj} = M_{inj}'BMx \text{~~for injection measurements} \label{injmat}
\end{align}
$B$ is the diagonal matrix of susceptances of lines in $\cal E$. We arrange the rows in $M$ such that the top $|z_f|$ rows represent the lines with flow measurements. Matrix $T$, comprising of the top $|z_f|$ rows of a $n_{\cal E} \times n_{\cal E}$ identity matrix, selects these measured flows. For ease of notation and analysis in later sections, we pad trailing zeros to vector $z_f$ and make it of length $n_{\cal E}$. Similarly, we pad trailing all-zero rows to $T$ to make it a diagonal square matrix of dimension $n_{\cal E}$. $M_{inj}$ on the other hand consists of the columns of $M$ that correspond to the nodes with injection measurements. The optimal state vector estimate $\hat{x}$ is given by minimizing the residual $\Sigma_e^{-.5}\|z-H\hat{x}\|_2$. If the minimum residual does not satisfy a tolerance threshold, bad-data detection flags turn on and data correction is done by the estimator. The overall scheme of topology and state estimation processes followed by bad-data detection and correction is called generalized state estimation (GSE) \cite{bookdc} as illustrated in Figure \ref{estimator}.
\begin{figure}
\centering
\includegraphics[width=0.38\textwidth, height=0.22\textwidth]{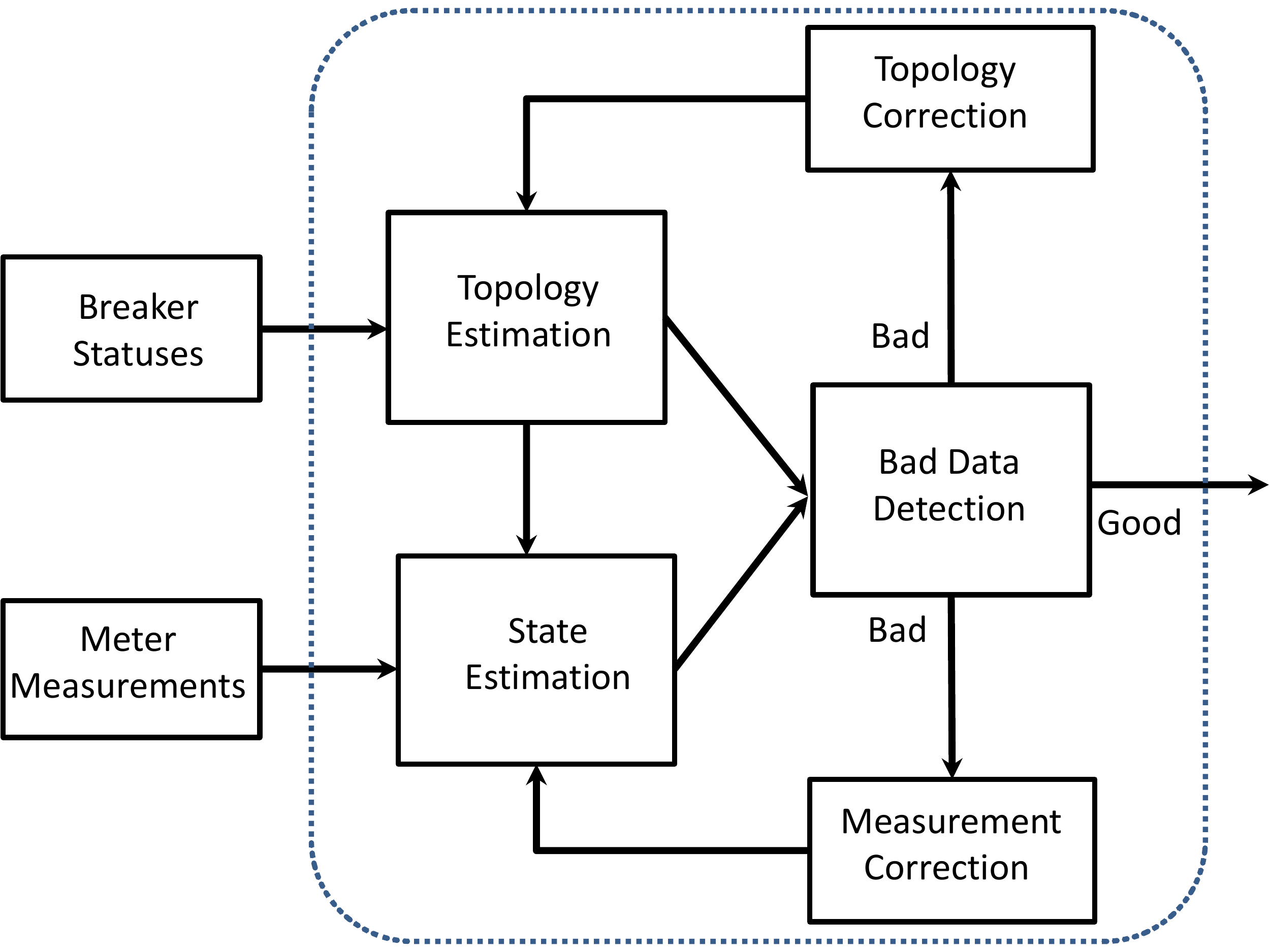}
\squeezeup
\caption{Generalized State Estimator for a power system}
\squeezeup
\label{estimator}
\end{figure}

\textbf{Attack Model:} We assume that the adversary is agnostic and has no information on the current system state $x$ or line susceptance matrix $B$. For attack, the adversary changes the breaker statuses on some lines. The new breaker status matrix, after attack, is denoted by $D-D_a$ where diagonal matrix $D_a$ has a value of $1$ for attacked breakers. Similarly, the available flow measurements after jamming are represented by $T-T_a$, with diagonal matrix $T_a$ having a value of $1$ corresponding to jammed flows. Let the new state vector estimated after the breaker-jammer attack be denoted by $x+c$, where $c \neq 0 $ denotes the change. Note that if the flow measurement on a line is not jammed, its value remains the same following the attack. Using (\ref{flowmat}), we have
\begin{align}
(T-T_a)BMx = (T-T_a)BM(x+c)\nonumber\\
\Rightarrow (T-T_a)BMc = 0 \label{flowcond}
\end{align}
It follows immediately that if the breaker status on the $r^{th}$ line with flow measurement is changed ($D_a(r,r) = 1$), to avoid detection, its flow measurement needs to be jammed as well ($T_a(r,r) = T (r,r) =1$). Thus,
\begin{align}
D_a(T-T_a) = 0 \label{breakjam}
\end{align}
Consider the injection measurements ($z_{inj}$) now, which are not changed during the attack. The breaker attack leads to removal of lines marked as open from Equation (\ref{injmat}), resulting in the following modification.
\begin{align}
z_{inj} = M_{inj}'BMx  &= M_{inj}'(D-D_a)BM(x+c)\nonumber\\
\Rightarrow~M_{inj}'D_aBMx &= M_{inj}'(D-D_a)BMc\nonumber\\
(\text{~Using~} (\ref{flowcond})) &= M_{inj}'(I-D_a)(I- T+T_a)BMc\nonumber\\
(\text{~Using~} (\ref{breakjam})) &= M_{inj}'(I-(D_a+T-T_a))BMc\label{injcond}
\end{align}
Equation (\ref{injcond}) thus states that after the 'breaker-jammer' attack, for each injection measurement, the sum of original flows contributed by lines with attacked breakers (left side) needs to be accommodated by changes in estimated flows on lines (connected to the same bus) whose breakers are intact but actual flow measurements are not received (right side). Finally, for unique state estimation following the adversarial attack (with one bus considered reference bus with phase angle $0$) we need
\begin{align}
\squeezeup
rank(\setlength{\arraycolsep}{2pt} \renewcommand{\arraystretch}{0.8}\begin{bmatrix} (T-T_a)\\ M_{inj}'(D-D_a)\end{bmatrix}BM) = n_{\cal V}-1\label{rank}
\squeezeup
\end{align}

The necessary conditions for a successful 'breaker-jammer' attack that results in a change in estimated state vector consists of equations (\ref{flowcond}), (\ref{breakjam}), (\ref{injcond}), and (\ref{rank}). In the next section, we describe a graph coloring based  analysis of the necessary and sufficient conditions and use it to discuss design of optimal attacks of our regime.

%An optimal attack of this regime is one that satisfies these conditions with minimum number of breaker status changes (considering the fact that changing a status is significantly more resource draining than measurement jamming). If multiple attacks are possible using the minimum possible number of breaker changes, we select as optimal the attack which requires lesser number of flow measurement jams.

\section{Attack Analysis based on Graph Coloring}
\label{sec:coloring}
For our graph coloring based analysis, we use the following coloring scheme: \textit{for any change $c$ in the estimated state vector, neighboring buses with same value in $c$ are given same color.} Using this, we now discuss a permissible graph coloring corresponding to the requirements of a feasible attack discussed in the previous section. Equation (\ref{flowcond}) states that if the flow on line $(ab)$ between buses $a$ and $b$ is not jammed, $c(a) = c(b)$ (same color in our scheme). Thus, \textit{\textbf{a set of buses connected through lines with available flow measurements (not jammed) has the same color.}}

This implies that the grid buses, following a feasible attack, can be divided into groups, each group having a distinct color. The lines between buses of different groups do not carry any flow measurement or are jammed by the adversary. A test example is illustrated in Figure \ref{fig:graphcoloring}. Observe the buses with injection measurements, that are not corrupted by the adversary. For an interior bus $d$, (all neighboring nodes have the same color as itself), the right side of Equation (\ref{injcond}) equates to zero. The left side becomes equal to zero, under normal operating conditions, if breakers on lines connected to bus $d$ are not attacked. Thus, we have \textit{\textbf{a feasible graph coloring has lines with attacked breakers connected to boundary buses.}} A boundary bus is one that has neighboring buses of colors distinct from itself.

\begin{figure}
\centering
\includegraphics[width=0.42\textwidth, height=0.30\textwidth]{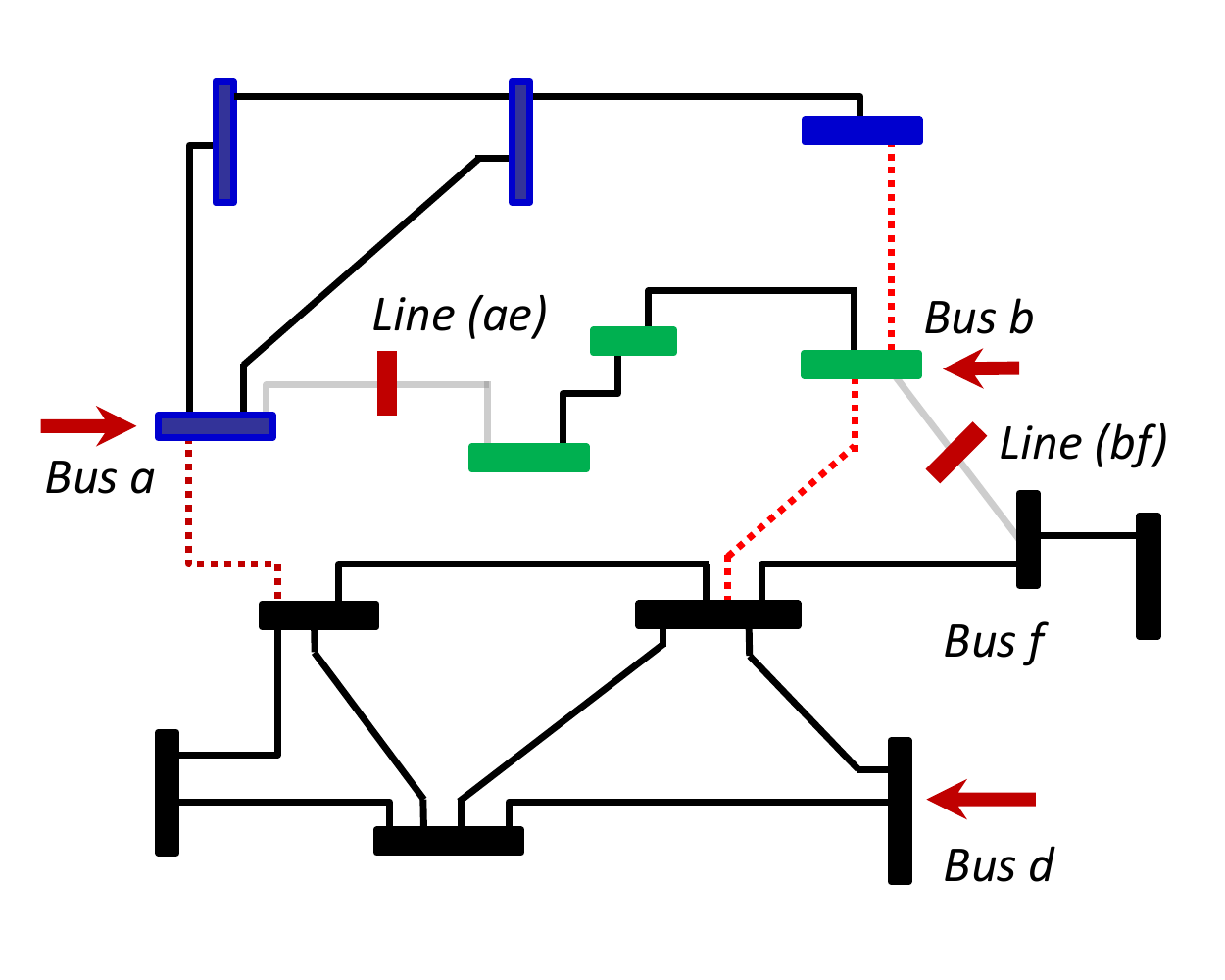}
\squeezeup
\caption{Feasible graph coloring scheme on IEEE $14$ bus system \cite{testsystem} with flow measurements on all lines and injection measurements at buses $a$, $b$ and $d$. The blue, green and black buses are divided into groups and have same value of change $c$ in estimated state vector. The dotted red lines represent jammed lines, solid black lines represent operational lines. The grey lines with red bars represent the lines $(ae)$ and $(bf)$ with attacked breakers.}
\squeezeup
\label{fig:graphcoloring}
\end{figure}

Now, consider the injection meter installed on any boundary bus. Such buses can exist in two configurations: a) connected to lines with attacked breaker (see bus $a$ in Figure \ref{fig:graphcoloring}) or b) connected to only lines with correct breaker statuses (node $b$ if line $(bf)$ did not have a breaker attack).
%In the first case (see node $a$ in Figure \ref{fig:graphcoloring}), the sum of flows due to $c$ on lines between the boundary bus to buses colored differently is equal to the flow on the line with compromised breaker. For the second case, the sum of flows on all lines connecting the concerned bus to differently colored buses of other groups is zero (using \ref{flowcond}). It is easily observed that for the second case, a feasible attack requires the boundary bus under consideration to have neighboring buses belonging to at least two different colored groups.
In either case, using (\ref{injcond}), we have:
\textit{\textbf{each injection measurement placed at a boundary bus provides one constraint relating the values of $c$ for neighboring differently colored buses.}}

For further analysis, we now use the coloring constraints highlighted in bold above to construct a reduced grid graph $\hat{\cal G}$ from $\cal G$ as follows:

\noindent \textbf{1}. In each colored group, club boundary buses without injection measurements with all interior buses into one 'supernode' of that color. Make boundary buses with injection measurements into supernodes with the same color. Connect supernodes of same color with artificial lines of zero susceptance.

\noindent \textbf{2}. For each line with intact breaker between two buses of different colors, create a line of same impedance between their corresponding 'supernodes'. Remove supernodes connected only to other supernodes of same color.

\noindent \textbf{3}. Make injection measurements on supernodes equal to the sum of original flows on lines with attacked breakers connected to them (positive for inflow, negative for outflow). If no incident line has attacked breaker, make the injection equal to $0$.

\begin{figure}
\centering
\includegraphics[width=0.33\textwidth, height=0.22\textwidth]{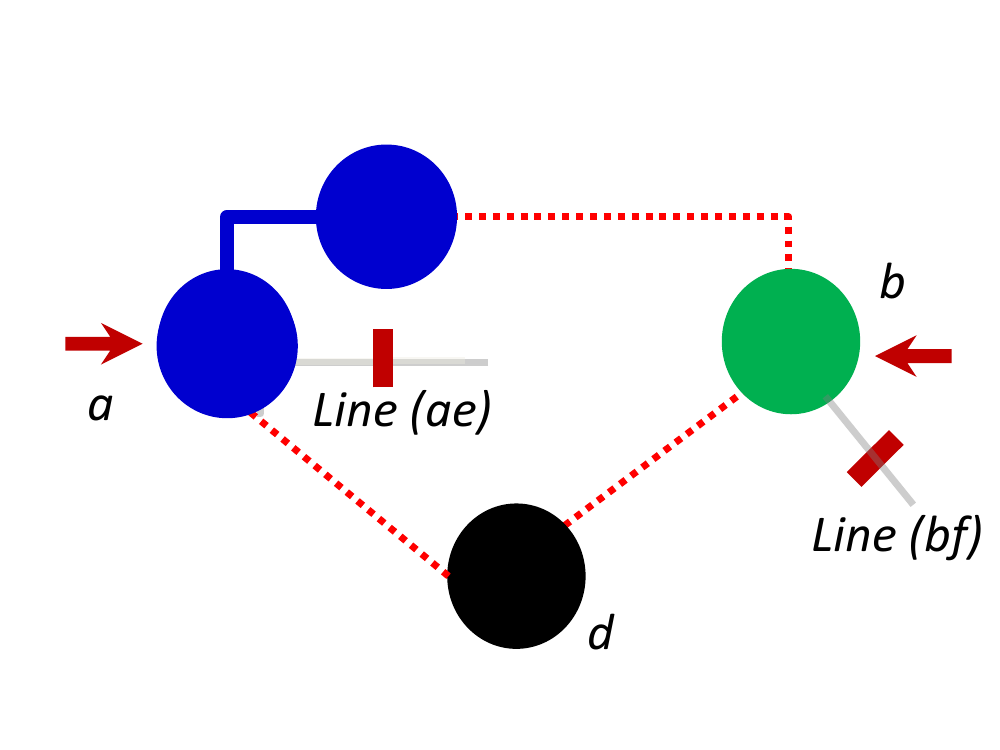}
\squeezeup
\caption{Reduced graph construction for the test case given in Figure \ref{fig:graphcoloring}. The blue, green and black solid circles represent super nodes for buses $a$, $b$ and $d$ respectively. The flow on the dotted red lines are not measured after attack. The grey lines with red bars represent lines with attacked breakers, that influence the injections at supernodes $a$ and $b$.}
\squeezeup
\label{fig:supernode}
\end{figure}

Figure \ref{fig:supernode} illustrates the reduced graph construction for the example in Figure \ref{fig:graphcoloring}. Note that in the reduced graph $\hat{\cal G}$, original lines between buses of same color are removed. The included lines exist between buses of different colors and have jammed or unavailable flow measurements. Similarly, injection measurement relation (\ref{injcond}) at interior nodes are trivially satisfied by $c$ and are ignored. The reduced system, thus, only includes constraints from boundary injection measurements that are similar in form to equation (\ref{injcond}) as shown below:
\begin{align}
\sum_{b: (ab) \in \hat{\cal E}}\hat{B}_{ab}(\hat{c}_a - \hat{c}_b) &= \hat{z}_{inj}(a) \label{reducedinjcond}
\end{align}
Here, $a$ and $b$ are supernodes of different colors. The numeric value for the color of supernode $a$ is given by $\hat{c}_a$ (not the $a^{th}$ entry in  $\hat{c}$). $\hat{B}$ and $\hat{\cal E}$ are the susceptance matrix and edge set corresponding to the reduced graph $\hat{\cal G}$. $\hat{z}_{inj}(a)$ denotes the injection measurement on supernode $a$ with value given by Step $3$ in the reduced graph construction. Note that equation (\ref{reducedinjcond}) for the injection measurements involves rows of the susceptance weighted Laplacian matrix for $\hat{\cal G}$. A unique solution of $\hat{c}$ for $\hat{\cal G}$ in turn provides a uniquely estimated $c$ in $G$ after the adversarial attack. We now look at condition (\ref{rank}), necessary for unique state estimation after a feasible adversarial attack in terms of graph coloring. The reduced graph $\hat{\cal G}$ greatly simplifies our analysis here.

First, it is clear that each color must have at least one supernode or a neighboring supernode (of different color) with injection measurement. Otherwise the value of $\hat{c}$ for that color will not be in any injection constraint. This goes against uniqueness of state estimation. Note that the number of degrees of freedom in $\hat{c}$ (representing distinct values in $c$) is one less than the number of colors as one color denotes the reference phase change of $0$.  Using $\hat{\cal G}$, we prove the following result regarding permissible graph coloring for unique estimation.

\begin{theorem} \label{oneless}
Following a 'breaker-jammer' attack, the number of injection measurements at the boundary buses should be one less than the number of distinct colors in the grid buses.
\end{theorem}
\begin{proof}
Let the number of colored groups be $k$. Then the number of independent entries in $\hat{c}$ is $k-1$ (one entry being $0$). The total number of linear constraints involving the numeric values in $\hat{c}$ is equal to the number of injection measurements at the supernodes in $\hat{\cal G}$. For unique state estimation, the number of injection measurements should thus be greater than or equal to $k-1$. We now show that exactly $k-1$ injection measurements are needed to get a solution to state estimation. Consider the reduced graph $\hat{\cal G}$. For real valued line susceptances and for cases where the supernodes having injection measurements do not form a closed ring with no additional branches (see Figure \ref{fig:supernode}), the rank of $k-1$ rows is $k-1$ and we have unique state estimation. If the reduced graph $\hat{\cal G}$ contains a closed ring of supernodes with injection measurements, then the measurements will represent the entire susceptance weighted graph Laplacian of the ring, that is rank deficient. However, the real valued entries in $\hat(z)_{inj}$ that exist on the right side of (\ref{reducedinjcond}) and are derived from flows on lines with attacked breakers, will not cancel out under normal operating conditions. Further, the adversary designing the attack is unaware of the current system state and will be able to determine if they do. Hence the $k-1$ injections measurements constraints will be linearly independent (the adversary will expect this under normal operations). This gives an unique $\hat{c}$ and $c$ for a $k$ distinct colored grid graph.
\end{proof}

To summarize, the highlighted statements and Theorem \ref{oneless} provide the necessary and sufficient conditions for a feasible 'breaker-jammer attack' under our graph-coloring scheme. In the next section, we show that the graph coloring approach proves a surprising result that simplifies the design of an optimal attack.

\section{Design of Optimal Attack}
\label{sec:design}
We call an 'breaker-jammer' feasible attack optimal if it requires minimum number of breaker status changes (considering the fact that doing so is significantly more resource draining than measurement jamming). If multiple attacks are possible using the minimum number of breaker changes, we select as optimal the attack that requires the least number of flow measurement jams. Using the reduced graph $\hat{\cal G}$, we present the following result for the minimum number of breaker changes needed for a feasible attack under normal operating conditions (non-zero real-valued bus susceptances and line flows that are distinct for different grid elements).

\begin{theorem} \label{oneenough}
If a feasible attack can be designed with $k$ breaker status changes, then a feasible attack exists such that all but one breaker statuses are changed back to their original operational state ($1$), while keeping their line flow measurements jammed.
\end{theorem}
\begin{proof}
Construct the reduced graph $\hat{\cal G}$ with its colored supernodes for the feasible attack with $k$ breakers and necessary flow measurement jams. Let the number of colors in state estimation change $c$ be $r+1$. The length of $\hat{c}$ is then $r+1$. By Theorem \ref{oneless}, there are $r$ injection measurements at the supernodes that provide constraint equations listed in (\ref{reducedinjcond}). If we revert the breaker status of an attacked line back to $1$ while keeping its flow measurement jammed, the only change in any constraint equation (\ref{reducedinjcond}) involving that line will be that the injection measurement on the incident node (entry in $\hat{z}$) will become $0$. Since all but one breakers are brought back to the operational state, at least one injection measurement in $\hat{z}$ will still remain non-zero and the $r$ constraint equations will still have linear inndependence. Thus, state estimation will result in a different but non-zero $\hat{c}$, leading to a feasible attack.
\end{proof}
For example, consider the case in Figure \ref{fig:graphcoloring} where two breaker statuses are attacked. If the breaker status on line $(bf)$ is changed back to $1$ while keeping the flow measurement jammed, the new reduced graph that will be derived is given in Figure \ref{fig:supernode1}. As mentioned in Theorem \ref{oneenough}, the coloring scheme is still feasible and a non-zero change in state estimation results.
\begin{figure}
\centering
\includegraphics[width=0.33\textwidth, height=0.22\textwidth]{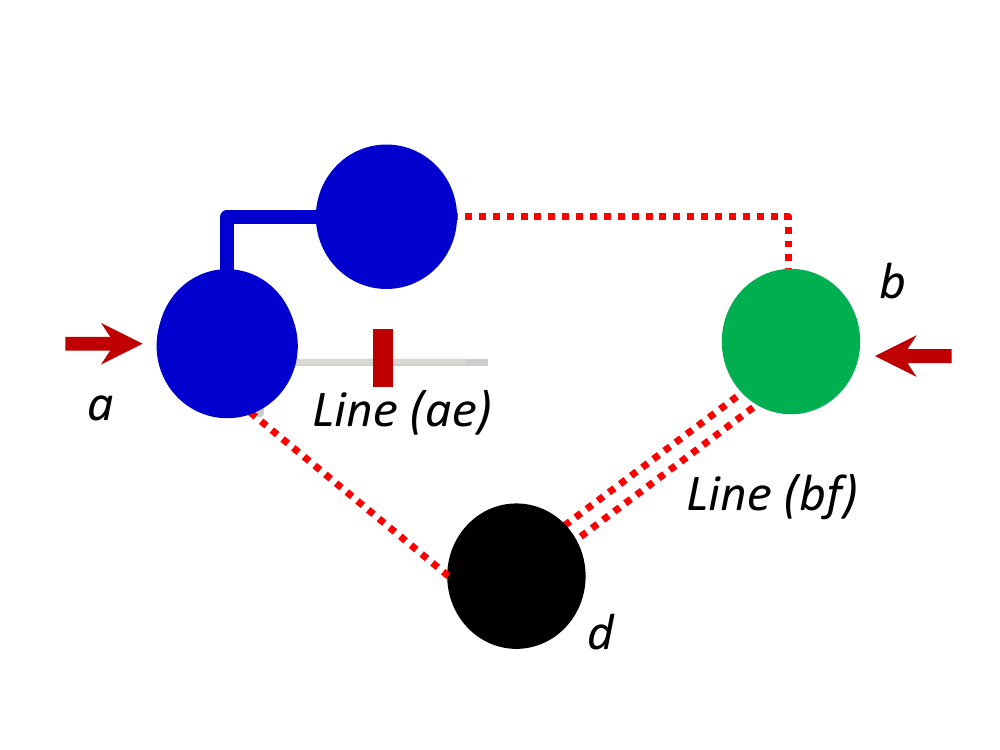}
\squeezeup
\caption{Reduced graph construction for IEEE $14$ bus case given in Figure \ref{fig:graphcoloring}, but with line $(bf)$ being changed to a dotted red line. The blue, green and black solid circles represent super nodes for buses $a$, $b$ and $d$ respectively. The flow on the dotted red lines are not received. The only grey line with red bar represents the line $(ae)$ with an attacked breaker.}
\squeezeup
\squeezeup
\label{fig:supernode1}
\end{figure}

This is a very significant result and simplifies the search for an optimal attack greatly. Since one breaker change is sufficient, the adversary can select each line in turn ($n_{\cal E}$ iterations), attack its breaker (change the corresponding entry in diagonal $D_a$ to $1$) and determine the flow measurements that need to be jammed (given by diagonal $T_a$) to conduct a feasible attack. The breaker change that requires the minimum number of measurement jams (or maximally sparse $T_a$) will then give the optimal attack. The selection of jammed measurements, after fixing $D_a$, is formulated as (\ref{opt_attack}).
\begin{align}
\squeezeup
\underset{c \neq \textbf{0},T_a}{\text{minimize}}~~~& \|T_a\|_0 \label{opt_attack} \\
\text{subject to}
~~~&T_a\text{~is diagonal,}~T_a\text{~satisfies (\ref{flowcond}), (\ref{breakjam}), (\ref{injcond})}\nonumber% \label{constraint}
\squeezeup
\end{align}
%\begin{tabular}{ccc}
%$\begin{aligned}
%&\underset{c \neq \textbf{0},T_a}{\text{min}}
%\|T_a\|_0 \text{~~~(P-1)}\\
%\text{s.t.~}
%&T_a\text{~is diagonal,} x \in \mathbb{R}^{n_{\cal V}}  \nonumber\\
%&T_a\text{~satisfies (\ref{flowcond}), (\ref{breakjam}), (\ref{injcond})}
%\end{aligned}$
%&$\xrightarrow[l_0-l_1]{\text{Relax}}$&
This is simplified in formulation (\ref{opt_attack1}) where the jammed measurements (with $1$ on diagonal of $T_a$) are given by the non-zero entries in $TMc$. $l_0-l_1$ relaxation can be used to approximately solve (\ref{opt_attack1}). Since the adversary has no access to the actual state vector $x$, a random non-zero $x$ is used. Similarly, unavailable line susceptance $B$ are replaced with distinct real values. These replacements, under normal conditions, do not affect the optimal solution as they preserve the linear independence of injection constraints given in (\ref{reducedinjcond}).
\begin{align}
\squeezeup
\underset{c \neq \textbf{0}}{\text{minimize}}~~~&\|TMc\|_0 \label{opt_attack1}\\
\text{subject to}
~~~&M_{inj}'D_aBMx = M_{inj}'BMc\nonumber
\squeezeup
\end{align}
The rank constraint (\ref{rank}) is not included in the optimization framework and can be checked manually after determining $T_a$, for consistency.\\
\textbf{Experiments:} We simulate our attack model on IEEE $14, 30$ and $57$ bus test systems \cite{testsystem} and present averaged findings in Figure \ref{fig:topologyplot}. For each test system considered, we place flow measurements on all lines and injection measurements on a fraction of buses, selected randomly. To design a feasible attack involving a line, we change its breaker status and solve Problem (\ref{opt_attack1}) to jam flows measurements to prevent detection. This is repeated for each line to determine the optimal attack. In Figure \ref{fig:topologyplot}, note that the average number of flow measurements jammed increases with the number of injection measurements. This happens due to an increase in the number of injection constraints that require more measurement jams.% The initial fall occurs as very low placement of injection meters lead to a large number of lines without injection measurements at both ends. Breaker attacks on such lines do not change the estimated state vector $x$ and are thus excluded from the attack regime. On the other hand, very high placement of injection meters provide a higher number of injection constraints that lead to increasing flow measurement jams.
\begin{figure}
\centering
\includegraphics[width=0.44\textwidth, height=0.32\textwidth]{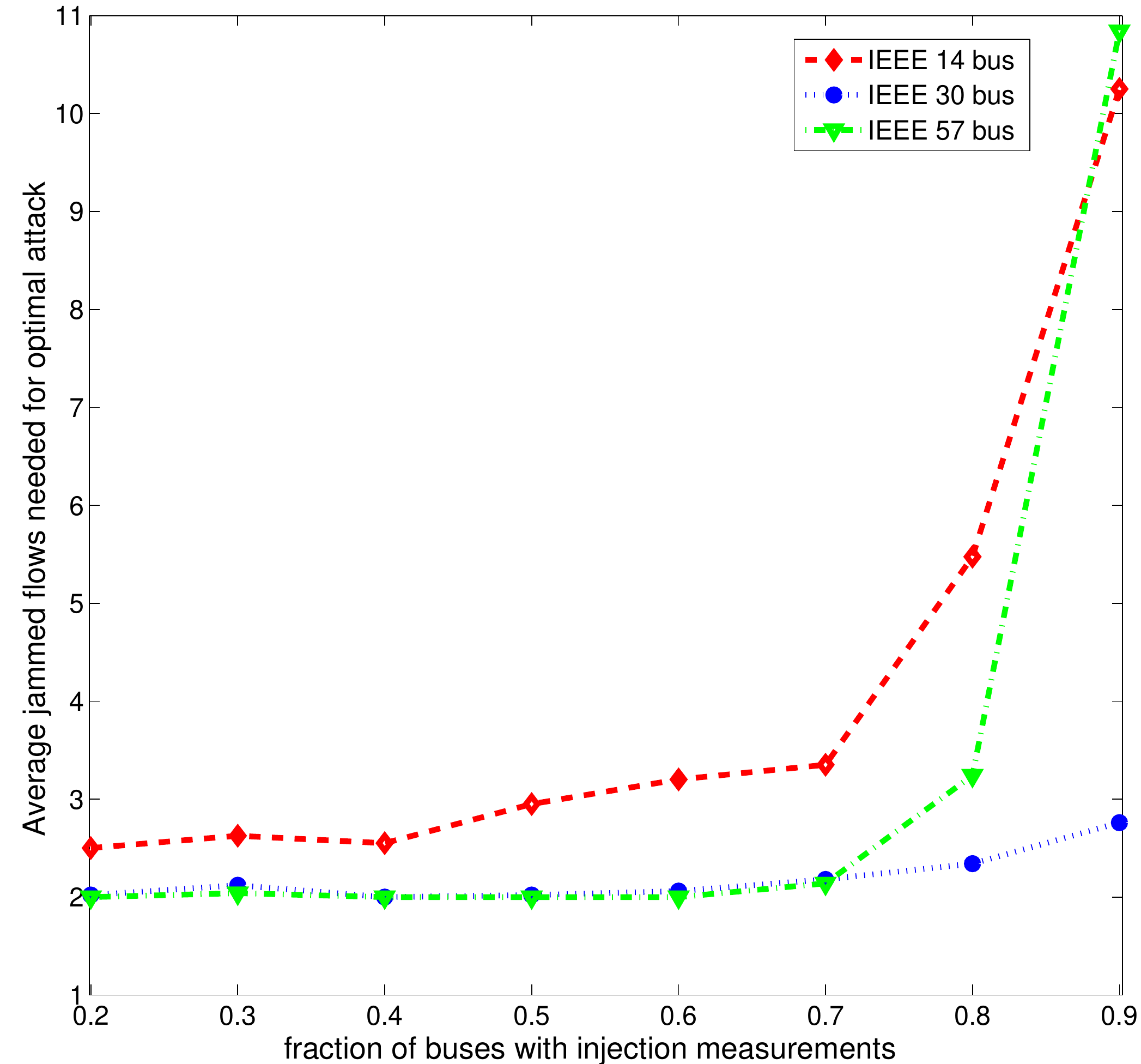}
\squeezeup
\caption{Average number of flow measurements jammed for optimal 'breaker-jammer' attacks on IEEE test systems. Injection measurements are placed on a fraction of buses (selected randomly) and flow measurements are placed on all lines.}
\squeezeup
\label{fig:topologyplot}
\end{figure}

\section{Conclusion}
\label{sec:conclusion}
In this paper, we study topology based cyber-attacks on power grids where an adversary changes the breaker statuses of operational lines and marks them as open. The adversary also jams flow measurements on certain lines to prevent detection at the state estimator. The attack framework is novel as it does not involve any injection of corrupted data into meters or knowledge of system parameters and current system state. Using lesser information and resource overhead than traditional data attacks, our attack regime explores attacks on systems where all meter data are protected from external manipulation. We discuss necessary and sufficient conditions for the existence of feasible attacks through a new graph-coloring approach. The most important result arising from our analysis is that optimal topology based attacks exist that require a single breaker status change. Finally, we discuss an optimization framework to select flow measurements that are jammed to prevent detection of the optimal attack. Its efficacy is presented through simulations on IEEE test cases. Designing protection schemes for our attack model is the focus of our current work.

\end{document}